\documentclass[twoside]{article}

\usepackage{algorithm}
\usepackage{algorithmic}
\usepackage{amsmath}
\usepackage{amssymb}
\usepackage{amsthm}
\usepackage{booktabs}
\usepackage{multirow}
\usepackage{caption}
\newtheorem{theorem}{Theorem}
\newtheorem{lemma}[theorem]{Lemma}

\newtheorem{definition}{Definition}
\usepackage{graphicx}
\usepackage{balance}

\usepackage[preprint]{aistats2026}
\usepackage[round]{natbib}

\begin{document}

%

%
\runningauthor{}

\twocolumn[

\aistatstitle{Local Search-based Individually Fair Clustering with Outliers}

\aistatsauthor{ Binita Maity$^{*}$, Shrutimoy Das$^{*}$, Anirban Dasgupta }

\aistatsaddress{\{binitamaity,shrutimoydas,anirbandg\}@iitgn.ac.in,\\Indian Institute of Technology Gandhinagar} ]

\begin{abstract}

In this paper, we present a local search-based algorithm for individually fair clustering in the presence of outliers. We consider the individual fairness definition proposed in \cite{jung}, which requires that each of the $n$ points in the dataset must have one of the $k$ centers within its $n/k$ nearest neighbors. However, if the dataset is known to contain outliers,  the set of fair centers obtained under this definition might be suboptimal for non-outlier points. In order to address this issue, we propose a method that discards a set of points marked as outliers and computes the set of fair centers for the remaining non-outlier points. Our method utilizes a randomized variant of local search, which makes it scalable to large datasets. We also provide an approximation guarantee of our method as well as a bound on the number of outliers discarded. Additionally, we demonstrate our claims experimentally on a set of real-world datasets.

\def\thefootnote{*}\makeatletter\def\Hy@Warning#1{}\makeatother\footnotetext{These authors contributed equally to this work}

\end{abstract}

\section{Introduction}

Machine learning based applications have become prevalent in various domains of life, be it in e-commerce, governance, banking, or healthcare. Ensuring that the algorithms used for training such systems are not susceptible to social biases has become an important area of research. According to a widely accepted definition of fairness, such algorithms must guarantee that no individual is at a disadvantage when availing these applications.
Achieving such {\em individual fairness} becomes more involved if there are outliers in the data. 
As an example, take the case of customer segmentation for bank loans, where customers of similar characteristics (demographics, transaction history, income,etc) are clustered together such that each segment gets offered similar loans. However, due to reasons such as errors in data gathering or abnormal spending habits of a few customers within a segment, the loans offered to that particular segment could be very different from the loans that would have been offered if such outliers were not present. That is, due to the presence of these outliers, the loans offered are not fair
to each individual. In an ideal scenario, the customers must be segmented in such a way that the loans offered to these customers are individually fair, across the segments, which requires that the outliers be excluded before segmentation.

This can be modeled as a problem of individually fair clustering after excluding outliers. Formally, given a set of $n$ points, which is known to contain a set of $m$ outliers, we want to cluster the $n-m$
points into $k$ clusters in an individually fair manner. \cite{jung} defines the notion of individual fairness in clustering by ensuring that each of the considered points must have a center within its $n/k$ neighbors. The algorithms for individually fair clustering, however,  are not suited to handle outliers. Naive application of these algorithms (in the dataset containing outliers) results in clusters that could potentially have a large fairness violation. That is, in the presence of outliers, the output of these algorithms becomes suboptimal for most points. Thus, the individually fair clusters must be computed without considering the outliers. On the other hand, applying a heuristic method to remove outliers and then computing fair clustering typically leads to an algorithm without any guarantees. 

Recently, \cite{binita} proposed a linear programming based approach to tackle this problem. However, their method is not scalable to large datasets due to the computational complexity of solving the linear program. Also, the authors did not provide any theoretical bound on the number of outliers that are discarded. In this paper, we propose a local search based algorithm. We give a novel center initialization algorithm that is able to discard outliers. Following this, we apply a local search algorithm. While local search based algorithms are known to be computationally expensive, we adopt a randomized
constrained local search algorithm \cite{bateni2024scalablealgorithmindividuallyfair} that makes our algorithm scalable to large datasets.

\paragraph{Clustering with outliers.} Clustering problems such as $k$-means, $k$-median, and $k$-center are widely used in real-world applications. Given a set of points, a similarity metric between two points, and a desired number of clusters, $k$, the goal is to partition these points into $k$ clusters, such that points within each cluster are similar. While $k$-means clustering is NP-hard even for $k = 3$, one of the most popular algorithms for clustering is Lloyd's heuristic, \cite{1056489}. In the $k$-means algorithm, it partitions the dataset into $k$ clusters to minimise the sum of squared distances to the corresponding cluster centers. However, in the real world, the dataset can be noisy, and noise in the dataset may degrade the quality of the clustering solution. To address this problem, several algorithms have been proposed, e.g. \cite{pmlr-v124-deshpande20a}, \cite{im2020fastnoiseremovalkmeans}.  \cite{Charikar2001AlgorithmsFF} proposed the first algorithm on $k$-median clustering with outliers, discarding additional points from the input set, identified as outliers, before applying the $k$-median algorithm. Similar to this, \cite{gupta2017localsearch} gave a local search-based method for $k$-means clustering in the presence of outliers in the data. 
\cite{10.5555/1347082.1347173} and \cite{krishnaswamy2018constantapproximationkmediankmeans} used a linear program to solve this problem, and hence suffer from a large runtime. \cite{huang2024nearlinear} gave a near-optimal time algorithm to tackle the problem. However, none of these algorithms considers a fairness metric. 


\paragraph{Individually fair clustering.}
The notion of individual fair clustering was introduced in \cite{jung}.  Several algorithms for individually fair clustering have been proposed in \cite{kleindessner2020notionindividualfairnessclustering,anderson2020distributional,mahabadi2020ind,pmlr-v151-vakilian22a,10.5555/3545946.359907}. Recently, \cite{bateni2024scalablealgorithmindividuallyfair} and  \cite{pmlr-v151-chhaya22a} proposed scalable methods for individually fair clustering. 




 \cite{mahabadi2020ind} gave a local search based algorithm for both $k$-median and $k$-means clustering while \cite{negahbani2021better} improved upon the guarantees on the approximation quality of the objective and the fairness by defining a linear program (LP) for this purpose as well as proposing rounding techniques for this LP.  A different LP formulation was also given in \cite{vakilian2022improved}. However, none of these works consider the setting where the dataset contains outliers. \cite{han2023approx} explored the individually fair $k$-center problem in the presence of outliers. The authors proposed an algorithm for minimizing the maximum fairness ratio of the non-outlier points. However, their work does not focus on minimizing the total cost. \cite{amagata,PDCAT} also look at the problem of fair $k$-center clustering in the presence of outliers, for group fairness. In this work, we explore the individually fair $k$-clustering problem in the presence of outliers using local search methods.

While \cite{binita} proposed an LP-based approach to solve this problem, their algorithm is not scalable and also fails to bound the number of outliers to be discarded. To address these issues, we propose a local search-based approach to solve individually fair $k$-means clustering in the presence of outliers in the dataset, where we provide a bound on the number of outliers to be discarded as well as an approximation guarantee on the clustering cost. 

Our main contributions can be enumerated as follows: 
\begin{enumerate}
    \item We present a novel center initialization algorithm, BaseCent (Algorithm \ref{alg:seeding_algo_main})  that discards fairness-based outliers, following which we refine the set of outliers and centers. 
    \item We present a local search-based method for individually fair $k$-means clustering when the dataset contains outliers.
    \item We provide a bound on the number of outliers discarded by our algorithm.
    \item We show that the proposed algorithm gives an $O(1)$ approximation to the cost of the optimal solution for the $(\gamma,k,m)$-fair clustering problem. 
    \item  We also validate the effectiveness of our algorithm empirically. 
\end{enumerate}

\section{Preliminaries}

Let $X$ be a finite set of points, $|X| = n$ and $d(p,S)$ be the distance between a point $p$ and its closest point in  $S,$ where $S \subseteq X.$ We consider $d(\cdot,\cdot)$ to be the Euclidean distance in this paper. If the set of centers $S$ is empty, then the $d(p,S)$ is $\infty.$ Also, let $d(p,q)$ be the distance between the points $p,q \in X.$ The aspect ratio of the instance is defined as $\Delta = \frac{max_{p,q} d(p, q)}{min_{p \neq q} d(p, q)} ,$ for any $p,q \in X.$ We define the $k$-means clustering cost as 
$cost(S, X) = \underset{p \in X}{\sum}  d(p, S)^{2}$ where $S$ is the set of centers such that $S \subseteq X$  and $k$ is the number  of clusters.

\begin{definition}[\emph{Fair radius $\delta(\cdot)$}] The fair radius for a point $v \in X$ is the radius of the ball containing the nearest $n/k$ neighbors  of $v.$ In other words, it is the distance of $v$ to its $n/k$-th nearest neighbour.
\end{definition}

A closely related term we will be using is a \emph{fair center.} A point $s \in S$ is a fair center for a point $p \in X$ if $s$ is a center and $d(p,s) \leq \delta(p).$

\begin{definition}[\emph{Fair $k$-means clustering}]  Given the set of data points $X$, a distance function $d(\cdot,\cdot),$  and the  fair radius function $\delta(\cdot),$ the fair $k$- means clustering problem looks at minimizing the clustering cost such that the distance from a point $v \in X$ to its center is at most $\delta(v).$ 
The problem  can be defined as follows, 
\begin{align} \label{def:fair_k_clustering}
 \begin{split}
     \underset{S \subseteq X :|S| \leq k}{min} &\sum_{v \in X, u \in S} d(v,u)^2 \\
    & \text{s.t } d(v, S) \leq \delta(v), \forall v \in X
 \end{split}   
\end{align}
\end{definition}

We note here that the definition in \eqref{def:fair_k_clustering} satisfies the individual fairness notion of clustering. We will be using the phrases individually fair $k$-means clustering and fair $k$-means clustering in this paper interchangeably. In this paper, we look at the problem of individually fair $k$-means clustering when the dataset is known to contain outliers. We consider two types of outliers in this paper: \emph{fairness-based outliers} are the points for which no fair centers can be assigned even after relaxation of their fair radii, while \emph{cost-based outliers} are the points that are far away from a given set  of centers.



\begin{definition}[\emph{$(\gamma, k, m)$- fair means clustering excluding outliers}]
Given a set of points $X$ in a metric space $(X,d).$ Let $Z$ denote the set of outliers  such that $|Z| \leq m.$ Then, a clustering of the points in $X \setminus Z$ with the set of centers $S \subseteq X \setminus Z$ is $(\gamma, k, m)$- fair if for all  $p \in X \setminus Z,$ we have $d(p,S) \leq \gamma \delta(p),$ for an appropriately chosen $\gamma.$
The problem of $(\gamma, k, m)$- fair $k$-means clustering can be defined as
\begin{align}
 \begin{split}
     \underset{\underset{Z,|Z| \leq m}{S \subseteq X\setminus Z :|S| \leq k}}{min} &\underset{v \in X\setminus Z, u \in S}{\sum}  d(v,u)^2 \\
    & \text{s.t } d(v, S) \leq \gamma \delta(v), \forall v \in X\setminus Z.
 \end{split}   
\end{align}
 

\end{definition}

    

    

    

\section{Proposed Algorithm}
Here, we present a local search-based algorithm for the $(\gamma,k,m)$-fair $k$-means clustering. Our algorithm has broadly two components: initializing the set of centers, $S_0,$  and refining this set of centers using a local search-based method. An important concept used in our algorithm is that of \emph{anchor zones}, which was introduced in \cite{bateni2024scalablealgorithmindividuallyfair}. We define the concepts of anchor points and anchor zones in the following discussion.

\cite{bateni2024scalablealgorithmindividuallyfair} proposed a \emph{seeding} algorithm for initializing a subset of centers $S_0, |S_0| \leq k,$ before running the local search algorithm. This seeding algorithm partitions the points into disjoint sets, each set being a ball around each $p \in S_0.$ The local search method proposed in \cite{bateni2024scalablealgorithmindividuallyfair} ensures that each of these disjoint sets is assigned at least one fair center. We refer to any point $p \in S_0$ as an \emph{anchor point.}
Also, for any $p \in S_0,$ we refer to  a ball $B(p)$ of radius $\gamma \delta(p)$ around $p$ as the \emph{anchor zone} for $p,$ for a suitable chosen  relaxation parameter $\gamma.$ Since every point in $X$ is assigned an anchor zone and each anchor zone is assigned at least one cluster center, these constraints enable us to bound the fairness violation of each point (relaxation of the fair radius required such that a fair center is assigned to it).

The seeding algorithm proposed in that paper returns infeasible if the size of $S_0$ becomes more than $k.$ Thus, if the dataset contains outliers, their algorithm may denote the problem as infeasible and the fair $k$-clusters for the non-outlier points will not be computed. Thus, we propose a novel initialization algorithm, \emph{BaseCent} (Algorithm \ref{alg:seeding_algo_main}), that computes the initial set of points $S_0,$ and also discards a set of points as outliers. Since, the set of points discarded at this step will not be considered for assigning fair $k$-means centers, these outliers are refered to as the fairness-based outliers.

\def\F{{\cal F}}
Let $\F_{k, m}$ be the set of all feasible instances of $(\gamma,k, m)$-fair $k$ clustering with $m$ outliers. We assume that we have an instance from $\F_{k,m}$. One possible instance where the $(\gamma,k,m)$ fair $k$-means clustering problem could become is infeasible when the points in $X \setminus Z$ requires more than $k+m$ anchor points.

\subsection{Initializing $S,$ the set of centers}
We discuss the BaseCent algorithm in detail. This algorithm, Algorithm \ref{alg:seeding_algo_main}, computes the initial subset of centers, $S_0$(the anchor points), and discards a set of fairness-based outliers, $Z_0.$ The motivation for this initialization algorithm has been drawn from the seeding algorithm in \cite{bateni2024scalablealgorithmindividuallyfair}. However, \cite{bateni2024scalablealgorithmindividuallyfair} does not consider the presence of outliers in the dataset.

Given the set of points $X$ and a relaxation parameter $\gamma,$  we keep growing the set of centers, $S,$ as long as there is a point $v \in X$ such that $d(v,S) > \gamma\delta(v).$ This process continues until  $m$ points remain. This is because the points are added to $S$ (or considered for adding to $S$) in order of their increasing fair radii. Thus, the points with the $m$ largest fair radii are discarded as the initial set of outliers, termed as fairness-based outliers.

In Lemma \ref{lem:gamma_plus_2}, we show that if we remove an additional set of $m$ cost-based outliers, then the $n-2m$ points, in a feasible solution of $\F_{k,m},$ can be covered by at most $k$ anchor zones with an increased relaxation parameter of  $\gamma' = \gamma+2.$ We refer to the previous anchor zones as $\gamma$-anchor zones and the new ones as the $\gamma'$-anchor zones. After discarding the fairness-based outliers, $Z_0,$ the remaining $n-m$ points might still have outliers, hence, we might require more than $k$ anchor zones to cover these points. Suppose we create $k+r$ anchor zones for the $n-m$ non-outlier points. Then there exists a feasible solution that discards an additional set of $m$ points as cost-based outliers and  covers the $n-2m$ non-outlier points using $k$ anchor zones with $\gamma' = \gamma+2$  fair radius relaxation. Note that we require at most  $k$ anchor zones for a feasible solution to  $\F_{k,m}$ since each anchor zone is assigned at least one fair center. However, we do not have an exact algorithm that can guarantee that we find these anchor zones. So, we apply a greedy set cover algorithm for computing these $k$ anchor zones.

Following step \ref{alg:step6} of the algorithm, we have $k+r$ $\gamma$-anchor zones ($|S_{rem}| = k+r$). For each anchor point, we create sets $P_i$'s where each point $x \in X'$ is assigned to $P_i$'s that are within $\gamma+2$ violation of $\delta(x).$ This implies that $P_i$'s are not necessarily disjoint.
We now have $k+r$ sets that fully cover $n-m$ points, out of which we know that there is a collection of $k$ sets that covers the $n-m$ points. We do not discard cost-based outliers in Algorithm \ref{alg:seeding_algo_main}. We know that the optimal set cover is bounded by $k.$ We run the greedy set cover algorithm on the sets $\{P_i\}_{i=1}^{|S_{rem}|}$, which gives at most $k\log n$ sets, i.e.,  $\mathcal{O}(k\log n)$ anchor zones with increased radius and their centers. The anchor points of these sets are the anchor points from $S$ that were used for creating the $P_i$'s. These anchor points are returned as $S_0.$




\begin{lemma}\label{lem:gamma_plus_2}
        Suppose  $n-m$ points are covered by $k+r$ $\gamma$-anchor zones $(0 \leq r \le m)$, with the $m$ points having the largest fair radii being discarded as fairness-based outliers. Then, there exists a set of $k$ $\gamma'$-anchor zones that covers $n-2m$ points, with $\gamma' = \gamma+2.$ 
\end{lemma}
\begin{proof}
    We consider a feasible solution $S$ of ${\cal F}_{k,m}$. Consider the $n-2m$ points that are not discarded as outliers (referred to as non-outliers) either in the feasible solution $S$ or via the $k+r$ anchor zones.  
    Now, let $c\in S$ be a fair center and let $P_c$ is the set of points that have been assigned to $c$. Let $p\in P_c$ be a point such that $p = \underset{p'\in P_c}{\text{argmax }} d(p',c),$. Since, $c$ is a fair center, $d(c,p) \leq \delta(p).$
    Again, since $p$ was assigned to a $\gamma$-anchor zone (one of the $k+r$ anchor zones) with the anchor point being $a$, $d(p,a) \leq \gamma \delta(p).$ Therefore, by triangular inequality, we have $d(c,a) \leq (\gamma+1) \delta(p).$ Now consider the ball $B(a)$ of radius $(\gamma+2)\delta(p))$ around $a.$ Then, all the points in $P_c$ are part of this ball.  So, centered at $a,$ if we allow any point $p'$ with $d(a,p') \leq (\gamma + 2) \delta(p)$ to be captured by the anchor zone at $a$ then any point in the fair cluster of $c$ will also be subsumed by these ``bigger" anchor zones.

    Since, given the larger anchor zones, the points can be clustered using $k$ fair centers and the anchor zones are disjoint, we can have at most $k$ such anchor zones. Note that it is not obvious which $k$ anchor zones to use out of the $k+r$ sets. 
\end{proof}

\begin{algorithm}[H]
\caption{\emph{BaseCent}}
\label{alg:seeding_algo_main}
\textbf{Input}: $X, \delta(.),\gamma$\\
\textbf{Output}: Set of anchor points $S_0$, set of fairness-based outliers $Z_0.$ 
\begin{algorithmic}[1] 
\STATE $ S \gets \phi, Z_0 \gets \phi, S_0 \gets \phi$
\WHILE{ $\exists p \in X : d(p,S) > \gamma \delta(p)$}
\STATE $p^* \gets argmin\{\delta(p'): p' \in \{p \in X| d(p,S)> \gamma \delta(p)\} \} $
\STATE $S \gets S \cup \{p^*\}$
\ENDWHILE

\STATE \label{alg:step6} Discard the last $m$ points to be assigned anchor zones as outliers ( as well as the anchor zones if those become empty after discarding the outliers). Let the discarded set be $Z_0$ and the remaining anchor points be $S_{rem}$

/*\textbf{Creating the sets with $(\gamma+2)$ relaxation}*/
\STATE $X' \gets X \setminus Z_0$
\STATE Let $P \gets \{P_i\}_{i=1}^{|S_{rem}|},$ where each $P_i$ contains the partition of points assigned to the anchor zone defined by the anchor point $S_{rem}^{i}.$
\STATE $P_i \gets \phi, \forall i \in [|S_{rem}|]$
\IF{$ |S_{rem}| > k$}
\FOR{ each $a \in X'$}
\FOR{ $i \in |S_{rem}|$}
\IF{$d(a,S^{i}_{rem}) < (\gamma+2)\delta(a)$}
\STATE $P_i \gets P_i \cup {a}$
\ENDIF
\ENDFOR
\ENDFOR
\\
/*\textbf{Apply Set Cover on the sets in $P.$}*/
\STATE $X'' \gets X'$
\WHILE{$X'' \neq \phi$}
\STATE Pick $j \in [|S_{rem}|] : j =  \underset{j}{\text{argmax}} |P_j \cap X''|$, the set which covers the maximum number of uncovered elements
\STATE $S_0 \gets S_0 \cup \{S^{j}_{rem}\}$
\STATE $X'' \gets X'' \setminus P_j$
\STATE $S_{rem} \gets S_{rem}\setminus\{S^{j}_{rem}\}$
\ENDWHILE
\ELSE
\STATE $S_0 \gets S_{rem}$
\ENDIF

\STATE \textbf{return} $S_0, Z_0$
\end{algorithmic}
\end{algorithm}

Lemma \ref{lem:gamma_plus_2} implies that after removing the initial set of $m$ points as  fairness-based outliers, the remaining $n-m$ points can be clustered such that each of the non-outlier points will have a fair center in its (relaxed) fair radius. However, there might still be outliers included in the set of $n-m$ non-outlier points. The main algorithm, Algorithm \ref{alg:ls}, takes care of these cost-based outliers.




\subsection{Local Search for fair $k$-means Clustering with Outliers}
The pseudo code for the main algorithm is given in Algorithm \ref{alg:ls}. It refines the initial subset of centers computed in the BaseCent algorithm and also discards additional points as cost-based outliers. We note here that the BaseCent algorithm discards fairness-based outliers while algorithm $2$  discards only cost-based outliers.  We select a new set of outliers (cost-based) from the set of non-outliers that satisfy fairness constraints. So the set of outliers output by Algorithm \ref{alg:ls} contains a union of both fairness-based as well as cost-based outliers. 

Algorithm \ref{alg:ls}, referred to as \emph{LSFO,} takes as input the set of points $X,$ the required number of clusters $k,$ the number of outliers known to be present in the dataset $m,$ the fair radius function $\delta(\cdot),$ the fairness relaxation parameter $\gamma,$ a constant $\epsilon >0,$ and a ball function $B(\cdot)$ that defines a ball of radius $\gamma+2$ around an anchor point $p \in S_0.$
The BaseCent algorithm outputs an initial set of anchor points ($S_0$) and outliers $Z_0.$ If $|S_0| < k$, then we randomly pick the remaining $k-|S_0|$ centers from the set of points $X \setminus \{Z_0 \cup S_0\}.$ 

The overall structure of Algorithm \ref{alg:ls} is motivated from \cite{gupta2017localsearch}. However, their algorithm employed local search while we employ the   constrained local search method proposed in \cite{bateni2024scalablealgorithmindividuallyfair} for improving scalability. The LSFO algorithm maintains a set of centers $S$ and a set of outliers $Z_0.$ While the local search method in \cite{gupta2017localsearch} initialized the centers arbitrarily, we utilize the BaseCent algorithm for this purpose. At each iteration, the current set of centers $S$ must satisfy the constraint that each anchor zone must have at least one center. The satisfaction of this constraint is ensured in step \ref{alg:LS_anchor_zone} of Algorithm \ref{alg:ls} and step \ref{alg:CLS_anchor_zone} of algorithm \ref{alg:CLS}.

Before discussing the steps of the algorithm, we define some of the notations used in Algorithms \ref{alg:ls} and \ref{alg:CLS}. The $cost(S,Z_0)$ is defined as the $k$-means cost for the points in $X_{fair}\setminus Z$ with respect to the set of centers $S.$  Also, define $cost(S,X)$ as the $k$-means cost for the set of points in $X,$ for any $X.$ The outliers function $outliers(S,Z_0)$ denotes the $m$ farthest points in the set $X_{fair} \setminus Z_0$ with respect to the centers in  $S.$ Also, $outliers(S)$ stores the set of outliers in the set $X_{fair}$ with respect to the set of centers $S.$
 
The LSFO runs in three stages. Given the current set of centers $S$ and the current set of outliers $Z_0,$ the first stage of the algorithm computes a locally optimal set of centers for the points in $X_{in}.$ We adopt the constrained local search method proposed in \cite{bateni2024scalablealgorithmindividuallyfair} to make this step scalable. In Algorithm \ref{alg:CLS}, we run the ConstrainedLS++ algorithm as long as there is a significant improvement in cost.

\begin{algorithm}[H]
\caption{\emph{Local Search for Fair $k$-means clustering with Outliers (LSFO)}}
\label{alg:ls}
\textbf{Input}: Datapoints $X$, number of centers: $k$, number of outliers : $m$, fair radius function $\delta(\cdot),$ relaxation parameter $\gamma,$ a constant $\epsilon,$ a ball function $B(\cdot)$\\
\textbf{Output}: Set of fair centers $S$, set of outliers $ Z_0$
\begin{algorithmic}[1] 
\STATE $S_0, Z_0 \gets BaseCent(X, \delta(\cdot), \gamma)$
\IF{ $|S_0| < k$}
\STATE $S \gets S_0 \cup  T\ \text{where}\ (T \subseteq X \setminus \{Z_0 \cup S_0\}, |T| = k-|S_{0}|)$
\ENDIF
\STATE $X_{fair} \gets X \setminus Z_0$
\STATE $\alpha \gets \infty$
\WHILE{$\alpha(1 - \epsilon/k) > cost(S, Z_0)$}
\STATE $\alpha \gets cost(S, Z_0)$
\STATE $X_{in} = X \setminus Z_0$\\
// Constrained Local Search with no  outliers
\STATE $S \gets ConstrainedLS++(X_{in},S_0,S, B(\cdot))$ 
\STATE $\bar{S} \gets S$
\STATE $\bar{Z} \gets Z_0$\\
// Discard $m$ additional outliers
\IF{$(1 - \epsilon/ k )\;cost(S, Z_0) > cost(S,Z_0 \cup outliers(S, Z_{0} ))$} 
\STATE $\bar{Z} \gets{Z_0 \cup outliers(S, Z_{0} )}$
\ENDIF\\
// Check if the set of centers can be refined with respect to $X_{fair}$ and discard additional outliers
\FOR{ each $u \in X_{fair}$ and each $v \in S$ }
\STATE \label{alg:LS_anchor_zone} $ V \gets \{v \in S| \forall x \in S_0 : (S \setminus \{v\} \cup \{u\}) \cap B(x) \neq \emptyset\}$
\STATE$ v^{\ast} \gets arg min_{v\in V} cost(S \cup \{u\} \setminus \{v\}, Z_0 \cup outliers(S \cup \{u\} \setminus \{v\}))$
\IF{$cost(S \cup \{u\} \setminus \{v^{\ast}\}, Z_0 \cup outliers(S \cup \{u\} \setminus \{v^{\ast}\})) < cost(\bar{S}, \bar{Z})$}
\STATE $\bar{S} \gets S \cup \{u\} \setminus \{v^{\ast}\}$
\STATE $\bar{Z} \gets Z_0 \cup outliers(S \cup \{u\} \setminus \{v^{\ast}\})$
\ENDIF
\ENDFOR
\IF{$(1 - \epsilon/ k )\;cost(S, Z_0) > cost(\bar{S},\bar{Z})$}
\STATE $S \gets \bar{S}$
\STATE $Z_0 \gets \bar{Z}$
\ENDIF
\ENDWHILE
\STATE \textbf{return} $S,Z_0$
\end{algorithmic}
\end{algorithm}

Once we obtain the set of centers from Algorithm \ref{alg:CLS}, in the second stage, we check to see if the discarding  further cost-based outliers leads to significant improvement in the objective. If there is an improvement in cost, we discard the additional set of cost-based outliers. Finally, in the third stage, we check whether the current set of centers in $S$ can be swapped with any of the points in $X_{fair}$ that leads to further improvement in the objective, after removing 
 $m$ additional cost-based outliers. After each iteration of the LSFO algorithm,  we maintain the best set of centers $S$ and outliers $Z_0$ computed until that iteration. The LSFO algorithm terminates if there is no significant improvement in the objective function.

\begin{algorithm}
\caption{ConstrainedLS++}
\label{alg:CLS}
\textbf{Input} : $ X, S_0, S, B(.)$\\
\textbf{Output}: Set of fair centers $S$
\begin{algorithmic}[1] 
\STATE $\alpha \gets \infty$
\WHILE{$\alpha(1 - \epsilon/k) > cost(S, X)$}
\STATE $\alpha \gets cost(S, X)$
\STATE Sample $p \in X $ with probability $\frac{cost(\{p\},S)}{ \sum _{q \in X}cost(\{q\},S) }$
\STATE \label{alg:CLS_anchor_zone} $ Q \gets \{q \in S| \forall x \in S_0 : (S \setminus \{q\} \cup \{p\}) \cap B(x) \neq \emptyset\}$
\STATE$ q^{\ast} \gets arg min_{q\in Q} cost(X,S \setminus \{q\} \cup \{p\} )$
\IF{$cost( X, S \setminus \{q^{\ast}\} \cup \{p\}) < cost( X, S)$}
\STATE$ S \gets S \setminus \{q^{\ast}\} \cup \{p\}$
\ENDIF
\ENDWHILE
\STATE \textbf{return} $S$
\end{algorithmic}
\end{algorithm}

\paragraph{Bounding the number of outliers}
The number of outliers discarded in Algorithm \ref{alg:ls}  can be bounded as follows. The initial set of outliers is discarded in the  BaseCent algorithm, which is  $m$  fairness-based outliers. 
Assuming that $\underset{p \neq q}{\text{min }} d(p,q) = 1,$ for any $p,q \in X_{fair}$ the worst possible cost is $O(n \Delta^2).$ In each iteration, the objective improves by a factor of at least $(1 - \frac{\epsilon}{k}).$ Thus, the LSFO algorithm achieves a cost of $1 \leq OPT$ in at most $O(\frac{\epsilon}{k}\log(n \Delta))$ iterations. In each iteration of Algorithm \ref{alg:ls}, at most $m$ additional outliers are discarded. 
 Hence, the total number of outliers discarded is $m+ \frac{mk}{\epsilon}\log(n \Delta).$

\paragraph{Runtime analysis}
Assuming the points are $d$ dimensional, the \textbf{while} loop in lines $2 -5$ in BaseCent  runs in time $O(ndk)$ and the greedy set cover runs in time $O(n k^2).$ Thus, the BaseCent algorithm runs in time $O(ndk + nk^2).$ As discussed previously, the \textbf{while} loops in algorithms \ref{alg:ls} and \ref{alg:CLS} runs in time $O(\frac{\epsilon}{k}\log(n \Delta)).$ Each iteration of Algorithm \ref{alg:CLS} has two steps : computing the distance of points to all centers ($O(ndk)$ time) and checking that the anchor zone constraint (at least one center in each anchor zone) is satisfied ($O(dk)$ time.) Thus, Algorithm \ref{alg:CLS} runs in time $O(\frac{ndk^2}{\epsilon}\log (n \Delta)).$ In stage $3$ of the LSFO algorithm, the maximum number of swaps is $O(kn)$ and each swaps requires $O(n)$ time for reassignment of points to the centers. Thus, the third stage requires $O(n^2 k)$ time. Since, the LSFO algorithm runs for $O(\frac{\epsilon}{k}\log(n \Delta))$ iterations, the total runtime is $O(ndk + nk^2 + (\frac{ndk^2}{\epsilon}\log (n\Delta) + n^2k)\frac{k}{\epsilon}\log(n \Delta) ) \approx O(\frac{ndk^3}{\epsilon^2} \log^2(n\Delta) + \frac{n^2k^2}{\epsilon}\log(n\Delta)).$

\section{Approximation Guarantee}
In order to analyze the approximation guarantees of Algorithm \ref{alg:ls}, we follow the framework of \cite{gupta2017localsearch}. However, in our case, since we have to consider the individual fairness constraints, the arguments have to satisfy the anchor zone constraints. Most of the lemmas are restatements from \cite{gupta2017localsearch})  and we  include lemmas and the proofs in the appendix. 

Our goal is to show that the set of centers computed by Algorithm \ref{alg:ls} is close to the optimal set of centers. In order to do this, we show that the cost of swapping a center from the set of solutions with an optimal center from the same anchor zone is not large. The details of this proof is very similar to the proof technique in \cite{gupta2017localsearch}, so we include this proof in the appendix. The construction of the permutation defined in \cite{gupta2017localsearch} will change due to the anchor zone constraints, which will be discussed in the appendix. We state the main approximation guarantee here.

\begin{theorem}
 Algorithm  \ref{alg:ls} is an $ O(1)$-approximation algorithm for any fixed $0 < \epsilon \leq \frac{1}{4}.$
\end{theorem}

 \section{Experimental Details}

 In this section, we implement the algorithms discussed in the previous sections and compare our results with different algorithms.  Similar to \cite{bateni2024scalablealgorithmindividuallyfair}, we have compared our results with \emph{Greedy} seeding algorithm as baseline (using the anchor points as the fair centers), \emph{ICML20} (\cite{pmlr-v119-mahabadi20a}), \emph{NeurIPS21,NeurIPS21Sparsify} (\cite{negahbani2021better}) and \emph{LSPP} (\cite{bateni2024scalablealgorithmindividuallyfair}).

\paragraph{Datasets} We have experimented on three datasets from the UCI Repository: Adult ($40k$ points), Bank($40k$ points), and Skin ($240k$ points) \cite{UCI}. In our experiments, we experimented on the full datasets, whereas for comparisons, we followed the dataset sizes used in \cite{bateni2024scalablealgorithmindividuallyfair}.

\paragraph{Metric} We experimented over two metrics : $k$-means cost and maximum bound ratio, which can be defined as $\rho = max_{p} \frac{d(p,S)}{\delta(p)}$,  where $S$ is the solution of the algorithm, and $\delta(\cdot)$ is the fair radius function.

\paragraph{\textbf{Introducing synthetic outliers}}

To compare with existing work, we experimented with $4000$ data points, except for LP-based experiments. There is a possibility that the sample may not contain outliers. We have artificially introduced outliers in the dataset by randomly sampling $1\%$ of the  sampled points and added uniform noise to each of the features of these points. For a feature $col,$ we add a noise sampled from $Uniform(0,col\_max),$ where $col\_max$ is the maximum feature value for $col$ over all the points. \cite{gupta2017localsearch}
used this method to generate artificial outliers.


\section{Results and Discussions}

In Table \ref{tab:comparison}, we compare the $k$-means cost and the maximum fairness violation, $\rho.$ Except LSPP and LSFO, the experiments for the other algorithms are conducted on $4000$ randomly sampled points, due to their high computational costs.
We observe that the $k$-means cost for our method (LSFO)  is much lower than each of the baselines. Also, the max bound ratio is less than LSPP in all the datasets.  In some of the cases, the $\rho$'s are better than ours. This could be due to the fact that the sampled points in these cases do not have points that violate the fair radius by large margins. The number of outliers discarded is $495$ for  Adult dataset, $707$ for the Bank dataset and $14000$ for the Skin dataset.

\begin{table}
\caption{Comparison over SOTA algorithms, given $k=10$. To make it comparable with the previous results, we run ICML20, NeurIPS21, and NeurIPS21Sparsify algorithms for $4000$ points. }
\label{tab:comparison}
\begin{tabular}{lcccc}
\toprule
Dataset & Algorithm & $k$-means cost & $\rho$ \\
\midrule
\multirow{5}{*}{adult} & Greedy & 1.56E+05 & 1.8 \\
 & ICML20 & 6.59E+04 & 1.4 \\
 & NeurIPS21 & 1.14E+05 & \textbf{1.2} \\
 & NeurIPS21Sparsify & 1.02E+05 & \textbf{1.2} \\
 & LSPP & 6.14E+04 & 1.4 \\
 & LSFO & \textbf{4.5E+04} & \textbf{1.2} \\
\midrule
\multirow{5}{*}{bank} & Greedy & 8.57E+04 & 1.9 \\
 & ICML20 & 3.23E+04 & 1.6 \\
 & NeurIPS21 & 5.68E+04 & \textbf{1.2} \\
 & NeurIPS21Sparsify & 5.70E+04 & \textbf{1.2} \\
 & LSPP & 3.02E+04 & 1.6 \\
 & LSFO & \textbf{1.0E+04} & 1.5 \\
\midrule
\multirow{5}{*}{skin} & Greedy & 1.80E+05 & 2.1 \\
 & ICML20 & 7.47E+04 & 1.8 \\
 & NeurIPS21 & 9.36E+04 & \textbf{1.1} \\
 & NeurIPS21Sparsify & 1.03E+05 & \textbf{1.1} \\
 & LSPP & 9.27E+04 & 3.1 \\
 & LSFO & \textbf{6.1E+04} & 1.91 \\
\bottomrule
\end{tabular}
\end{table}

We also compare against the  $LP$ based solution proposed in \cite{binita} in Table \ref{tab:my_label} on  the same set of $1000$ samples. We note here that the $LP$ based solution is not scalable to a larger number of points due to the $LP.$ We observe the $k$-means cost for LSFO is less than the cost for the $LP$ based method.
 
\begin{table}
\centering
\caption{Cost of Adult dataset with $1000$ points using both LSFO and LP \cite{binita}}
\label{tab:my_label}
\begin{tabular}{ccc}
\toprule
$k$ & LSFO & LP \\
\midrule
5  & $\mathbf{1.88E+03}$ & $2.01E+03$ \\
10 & $\mathbf{1.06E+03}$ & $1.47E+03$ \\
15 & $\mathbf{8.63E+02}$ & $1.15E+03$ \\
30 & $\mathbf{4.91E+02}$ & $8.35E+02$ \\
\bottomrule
\end{tabular}
\end{table}

In Table \ref{tab:time_m}, we state the time taken for varying number of centers on the full dataset for Adult and Bank datasets. With increase in the number of centers, the value of $\rho$ decreases as expected. Also, given that we set $m = 400,$ ($1\%$ of the dataset size), the number of outliers discarded is less than $2m$ in most cases. 


\begin{table}
\centering
\caption{\label{tab:time_m} Results on Adult and Bank Datasets}
\begin{tabular}{@{}cccccc@{}} 
\toprule
Dataset & $k$ & $m$ & $k$-means cost & $\rho$ & Time (sec) \\ 
\midrule
\multirow{5}{*}{Adult} 
 & 5  & 595 & $7.55E+04$ & 1.44 & 31 \\
 & 10 & 707 & $4.57E+04$ & 1.22 & 242 \\
 & 15 & 755 & $3.81E+04$ & 1.07 & 776 \\
 & 20 & 772 & $3.30E+04$ & 1.07 & 1552 \\
 & 30 & 905 & $2.60E+04$ & 1.05 & 3821 \\
\midrule
\multirow{5}{*}{Bank} 
 & 5  & 463 & $1.82E+04$ & 2.20 & 24.29 \\
 & 10 & 495 & $1.06E+04$ & 1.59 & 82.40 \\
 & 15 & 725 & $6.95E+03$ & 1.29 & 316.00 \\
 & 20 & 647 & $6.06E+03$ & 1.32 & 470.00 \\
 & 30 & 828 & $4.66E+03$ & 1.20 & 966.00 \\
\bottomrule
\end{tabular}
\end{table}


In Table \ref{tab:varying_gamma}, we vary the value of the relaxation parameter $\gamma$ and report the cost, $m,$ $\rho$ and time taken. It can be observed that none of these metrics are affected much by varying $\gamma,$ which is to be expected as the theoretical bounds for the run time and the number of outliers discarded are independent of $\gamma.$

\begin{table}
\centering
\caption{ \label{tab:varying_gamma}Results for Adult Dataset ($k=5$, $\epsilon=10^{-4}$)}
\begin{tabular}{ccccc}
\toprule
 $\gamma$ & $m$ & Cost & $\rho$ & Time (sec) \\
\midrule
1 & 595 & $7.32E+04$ & 1.44 & 34.53 \\
2 & 611 & $6.94E+04$ & 1.41 & 65.75 \\
3 & 592 & $7.44E+04$ & 1.44 & 38.34 \\
4 & 600 & $7.17E+04$ & 1.51 & 44.86 \\
5 & 588 & $7.35E+04$ & 1.50 & 36.69 \\
\bottomrule
\end{tabular}
\end{table}

In Table \ref{tab:varying_eps}, we plot the various evaluation metrics by varying the parameter $\epsilon.$ It can be observed that for $\epsilon =  1e-4,$ the time taken is much faster than for the other values of $\epsilon.$ Thus, we considered $\epsilon = 1e-4$ in all of our experiments.

\begin{table}
\centering
\caption{\label{tab:varying_eps}Results for Adult Dataset ($k=5$, $\gamma=3$) for different $\epsilon$ values}
\begin{tabular}{ccccc}
\toprule
$\epsilon$ & $m$ & Cost & $\rho$ & Time (sec) \\
\midrule
$1e{-5}$ & 758 & $6.73E+04$ & 1.46 & 53.11 \\
$1e{-4}$ & 595 & $7.21E+04$ & 1.46 & 4.38 \\
$1e{-3}$ & 677 & $7.00E+04$ & 1.42 & 50.34 \\
$1e{-2}$ & 643 & $7.18E+04$ & 1.46 & 39.73 \\
$1e{-1}$ & 592 & $7.33E+04$ & 1.45 & 36.60 \\
\bottomrule
\end{tabular}
\end{table}

\section{Conclusion}
In this paper, we present a local search based method for individually fair $k$-means clustering when the dataset contains outliers. We present a novel center initialization algorithm that discards fairness-based outliers, following which we refine the set of outliers and centers. The proposed algorithm gives an $O(1)$ approximation to the cost of the optimal solution for the $(\gamma,k,m)$ fair clustering problem. The empirical results further validate the effectiveness of our algorithm. One direction of future work would be to extend this algorithm to other definitions of fairness, such as balance.

\balance
\bibliography{References}

\newpage
\appendix


\section{Swap pair and capture}

In the discussions that follow, let $S$ denote the set of centers output by Algorithm \ref{alg:ls} and $S^*$ denote the optimal set of centers. Also, let $Z_0$ denote the set of outliers computed by algorithm \ref{alg:ls} and $Z^*$ be the optimal set of outliers.

\noindent\textbf{Local optimality} 
At termination of Algorithm \ref{alg:ls}, we are at a local optimum and the following conditions are satisfied:

\begin{enumerate}
    \item $cost(S, Z_0 \cup outliers(S)) - cost(S, Z_0) \geq - \frac{\epsilon}{k}cost(S, Z_0)$
\item for any $ u \in X_{fair}$ and $v \in S,$ we have $cost(S \cup \{u\} \setminus \{v\}, Z_0 \cup outliers(S \cup \{u\} \setminus \{v\}))) - cost(S, Z_0) \geq - \frac{\epsilon}{k} cost(S, Z_0)$

\end{enumerate}

Similar to \cite{gupta2017localsearch}, we define the notions of swap pairs and capture. It is to be noted that in our case, these notions must satisfy the anchor zones constraints as well.

\textbf{Capture:} For $a \in S$ and $ b \in S^*
$ we say that $a$ captures $b$ if $a$ and $b$ are in the same anchor zone, and  $$|N(a) \cap (N^{*}(b) \setminus Z_0)| >1/2|N^{*}(b) \setminus Z_0)|,$$\\
where  $N(a)$ is the set of points with $a \in S$ as their cluster center and $N^*(b)$ is the set of points with $b \in S^*$ as their cluster center.
Note that each $b \in S^*$ can be captured by at most one $a \in S$.

We define a set $P$ of swap pairs as follows.
\begin{itemize}
    \item  If $a \in S$ captures exactly one $b \in S^*$, then add $(a, b)$
to $P$.
\item  For each $a \in S$ that captures no  $b \in S^*$, add the pair
$(a, b)$ to $P$ for any  $b \in S^*$
such that \begin{enumerate}
    \item $b$ is not already
included in a swap pair
\item each $a \in S$ is involved in at most two swap pairs
\item $a$ and $b$  belong to the same anchor zones.
\end{enumerate}
\end{itemize}

\textbf{Properties of swap pairs.} The key to this definition of swap pairs is ensuring that it has the following properties.
\begin{itemize}
    \item Each $b \in S^*$ is involved in exactly one swap pair;
    \item  Each $a \in S$ is involved in at most two swap pairs;
    \item If $ (a, b) \in P$ then $a$ captures no $ b' \in S^* \setminus \{b\}.$
\end{itemize}

Next we show that in each anchor zone,  there are enough centers in $S$ that capture no center in $S^*$ such that every $b \in S^*$ is involved in one swap pair. Let $\beta$ denote the number of centers in $S$ that capture exactly one center in $S^*.$ 
\begin{lemma}
Let $I$ be a specific anchor zone. Let $s \subseteq S$ and $ s^* \subseteq S^*$ belong to the same anchor zone $I$.
Let $\beta$ be the number of centers in $s$ that captures exactly one center in $s^*$.
There are at least $\frac{|s|- \beta}{2}$
centers in $s$ that capture no center in $s^*$.
\end{lemma}


\begin{proof}
    Let $t$ be the number of centers in $s$ that capture two or more centers in $s^*$, and $x$ be the number that capture none. We know that $2t + \beta \le |s| \le k$ since each center in $s^*$ can
be captured by at most one center in $s.$ 
This implies that $t  \leq \frac{k-\beta}{2}$. We further know that $t + x + \beta = |s|$ as this counts all centers in $s.$

Hence we have that $|s| = t + x + \beta \leq \frac{|s| - \beta}{2} + x + \beta.$
This implies $x \geq \frac{|s|- \beta}{2}$
proving the lemma.
\end{proof}

The cost of a swap is defined as $$cost(a, b) :=
cost(S \cup \{b\} \setminus \{a\}, Z_0 \cup outliers(S \cup \{b\} \setminus \{a\})).$$

Summing over all $k$ swap pairs, we have
$$ \sum_{(a,b) \in P} (cost(a, b) - cost(S, Z_0))
\geq - \epsilon cost(S, Z_0),$$ from local optimality. 
We show that,  $ \underset{(a,b) \in P}{\sum} cost(a, b)$ is bounded by $O(cost(S^*, Z^*))$, the optimal  solution’s cost. 

\textbf{Defining a permutation}
Similar to \cite{gupta2017localsearch}, we construct a permutation with the additional anchor zone constraints. We consider a permutation $\pi$ on the ordering of the points in $X \setminus (Z_0 \cup Z^*)$ such that the points belonging to $N^*_{b_i}$ come before $N^*_{b_{i+1}}$  for all $i = 1, 2 \ldots k-1,$ in the permutation. Furthermore,  for the set of points in $N^*_{b_i}, \forall i,$ we consider an ordering of the points such that all the points in  $N_{a_j}$ 
come before $N_{a_{j+1}}$ for $j = 1, 2,\ldots, k-1.$ Also, if the set of points belonging to $N^*_{b_r}, \ldots, N^*_{b_{r+p}} $ belong to the same anchor zone, then these points are in consecutive locations in the permutation.

Using this permutation, we  also define a mapping between the points. Let $\pi(v)$ denote a point that is diametrically opposite to $v$ in the above permutation, where $v, \pi(v) \in N^*_{b_i}.$ 

Using this permutation, we want to show an upper bound on  $cost(a,b).$
Consider the swap pair $(a,b).$ For the centers in $S \cup \{b\} \setminus \{a\},$ we define a suboptimal clustering of points in $X \setminus Z_0$ and discard $m$ additional points as outliers. Since this is a suboptimal assignment, the cost of this assignment will be larger than $cost(a,b).$ Also, the total number of outliers in this clustering will be $|Z_0|+m.$  For the set of points in $X ,$ the following cases arise  for the set of centers in $S \cup \{b\} \setminus \{a\}$ and the suboptimal clustering

\begin{itemize}
    \item Each $v \in Z_0 \setminus N^*_b$ and each $v \in Z^*$ is an outlier
    \item Each $v \in N^{*}_b$ is mapped to the center $b$
    \item Each $v \in N_{a',a' \neq a} \setminus (Z^* \cup N^*_b)$ is mapped to $a'$
    \item Each $v \in N_{a} \setminus (Z^* \cup N^*_b)$ is mapped to the center for
$\pi(v) $ in the solution obtained by algorithm \ref{alg:ls}. 

\end{itemize}
Next, we show that for each $v \in N_{a} \setminus (Z^* \cup N^*_b),$ the point $\pi(v)$ is mapped to $a', a \neq a'$ in the solution obtained by Algorithm \ref{alg:ls}. Also, since $v, \pi(v) \in N^{*}_{b'},$ both $v$ and $\pi(v)$ belong to the same anchor zone. Thus, the assignment of  $v$ to $a'$ is valid.

We restate lemma $8$ from \cite{gupta2017localsearch}.

\begin{lemma}
    For any point $v \in N_{a} \setminus (Z^* \cup N^*_b),$ we have $\pi(v) \in N_{a'}$ for some $a' \neq a.$
\end{lemma}
 
 \begin{proof}
By definition, $\pi$ is a permutation on points in $ X \setminus (Z_0 \cup Z^*).$ In particular, this implies that $\pi(v)$ is not an outlier either in algorithm \ref{alg:ls} or in the optimum. We know that $v \in N^{*}_{b'}$ for some $b' \neq b$ by assumption. Using the permutation $\pi$, $v$ is mapped to the diametrically opposite point $\pi(v)$ in $N^{*}_{b'}$. 
Knowing that $a$ does not capture any $b' \neq b$ by definition of swap pairs, it is the case that $\pi(v)$
must not be in $N_a$. Since $ \pi(v)$ is not an outlier, this implies
the lemma.
 \end{proof}

The structure of the proofs in \cite{gupta2017localsearch} is closely followed by the proof in this section, with a few adjustments made to handle the anchor zones constraints carefully.
\begin{lemma}
It is the case that 
\begin{align*}
\begin{split}
    \bigg(\sum_{v \in N^*_{b}} d(v, b)^2 - \sum_{v \in N^*_{b} \setminus Z_0}d(v,v^c)^2 \bigg) \\ \leq  cost(S^*, Z^*) - (1 - \frac{\epsilon}{k})cost(S, Z_0).
\end{split}
\end{align*}
\end{lemma}  

\begin{proof}
\begin{align*}
\begin{split}
    \sum_{(a,b) \in P} \biggl( \sum_{v \in N^*b} d(v, b)^2 -\sum_{v \in N^{*}b\setminus Z_0} d(v, v^c)^2 \biggr) \\ \leq cost(S^*, Z^*) - \sum_{(a,b) \in P} \sum_{v \in N^*b\setminus Z_0} d(v, v^c)^2 
\end{split}
\end{align*}
\\ \\
(each $b$ appears once in a swap pair in $ P$)
We can further bound this as:
\begin{align*}
\begin{split}
\leq cost(S^*,Z^*) - \sum_{v \in U \setminus Z_0}d(v, 
v^c)^2 + \sum_{v \in Z* \setminus Z_0} d(v, v^c)^2 \\   = cost(S^*, Z^*) - cost(S, Z_0) + \sum_{v \in Z^{*}\setminus Z_0}d(v, v^c)^2  \\ \leq cost(S^* Z^*) - cost(S, Z_0) +  \frac{\epsilon}{k}cost(S, Z_0)
\end{split}
\end{align*}
\end{proof}

\begin{lemma}For any positive real numbers $x, y, z$ and parameter $0 < \delta \leq 1$ it is the case that $ (x + y + z)^2 \leq \biggl( 1 + \frac{2}{\delta} \biggr) (x + y)^2 + (1 + 2 \delta)z^2.$
\end{lemma}
\begin{proof}
    \begin{align}
        \begin{split}
            (x + y + z)^2 \\ = (x + y)^2 + z^2 + 2z(x + y) \\
\leq \biggl( 1 + \frac{2}{\delta} \biggr) (x + y)^2 + (1 + 2\delta)z^2 \\
        \end{split}
    \end{align}
    (either $ x + y  \leq \delta z$ or $ x + y > \delta z)$
\end{proof}
\begin{lemma}
\begin{align*}
    \begin{split}
        \sum_{(a,b) \in P} \sum_{v \in N_{a}\setminus (Z^* \cup N^*_{b})}(d(v, \pi^{c}(v))^2-d(v, v^c)^2) \\ \leq \biggl(8 + \frac{16}{\delta} \biggr)cost(S^*, Z^*) + 4 \delta cost(S, Z_0)
    \end{split}
\end{align*}
 for any $0 < \delta \leq 1.$
\end{lemma}

\begin{proof}
Fix any $(a, b) \in P.$ Notice that
\begin{align*}
\begin{split}
    d(v, \pi^{c}(v)) \leq d(v, \pi(v)) + d(\pi(v), \pi^{c}(v)) \\ \leq d(v, b) + d(b, \pi(v)) + d(\pi(v), \pi^{c}(v))
\end{split}
\end{align*}
where both inequalities follow from the triangle inequality.

Consider any $v \in X \setminus (Z_0 \cup Z^*)$. For any such point $v, \pi(v)$ is
well defined. Notice that $d(v, b)+d(b, \pi(v))+d(\pi(v), \pi^{c}(v)) - d(v, v^{c}) \geq 0.$ This is because the first three terms form an upper bound on the distance from $v$ to a center in $S \setminus {v^{c}}$ and $ v$ is assigned to center $v^c$ in the clustering obtained by \ref{alg:ls} (the closest center to $v$ in $ S$).

\begin{align*}
    \begin{split}
    \sum_{(a,b) \in P} \sum_{v \in N_{a} \setminus (Z^* \setminus N^*_{b})}(d(v, \pi^{c}(v))^2 - d(v, v^c)^2) \\ \leq \sum_{(a,b) \in P} \sum_{v \in N_{a} \setminus (Z^* \cup N^*_{b})} \biggl(d(v, b) + d(b, \pi(v)) \\ +d(\pi(v), \pi^{c}(v))^2 -d(v, v^c)^2 \biggr) \\
\leq 2 \sum_{v \in X\setminus(Z^* \cup Z_0)} \biggl(d(v, b) + d(b, \pi(v)) \\ +d(\pi(v), \pi^{c}(v))^2 - d(v, v^c)^2 \biggr)
    \end{split}
\end{align*}

Continuing, we bound the sum as 
\begin{align*}
    \begin{split}
        \leq 2 \sum_{v \in X\setminus (Z^* \cup Z_0)} \biggl(
\biggl(1+\frac{2}{\delta}\biggr) d(v, b) + d(b, \pi(v))^2 + \\ (1 +2 \delta)d(\pi(v), \pi^{c} (v))^2 -d(v, v^c)^2\biggr) \\
\leq 2 \sum_{v \in X\setminus (Z^* \cup Z_0)}
\biggl( \biggl(2 + \frac{4}{\delta} \biggr) (d(v, b)^2 + d(b, \pi(v))^2  \\ +(1 + 2\delta)d(\pi(v), \pi^c(v))^2 -d(v, v^c)^2 \biggr)
    \end{split}
\end{align*}

To complete the lemma, consider the value of 
\begin{align*}
\sum_{v \in X\setminus (Z^* \cup Z_0)}  \biggl(d(\pi(v),\pi^c(v))^2 - d(v, v^c)^2 \biggr).
\end{align*} 
This summation is over all the points considered in the permutation $\pi$, and each point $v$ in the permutation is mapped to by exactly one other point that is diametrically opposite. Due to this, each point $v \in X \ (Z^* \cup Z_0) $ contributes $d(v, v^c)^2$ once in the first term and once in the second term.

\begin{align*}
    \begin{split}
        \sum_{v \in X\setminus (Z^*\cup Z_0)} \biggl(d(\pi(v), \pi^c(v))^2 - d(v, v^c)^2 \biggr)= 0.
    \end{split}
\end{align*}
This argument implies the following.
\begin{align}
    \begin{split}
        \leq 2 \sum_{v \in X\setminus (Z^* \cup Z_0)}
\biggl(2 + \frac{4}{\delta}\biggr) (d(v, b)^2 + d(b, \pi(v))^2) \\ +2 \sum_{v\in X\setminus (Z^* \cup Z_0)} 2\delta d(\pi(v), \pi^{c}(v))^2 \\ \leq
\biggl(8 + \frac{16}{\delta}\biggr) cost(S^*,Z^*) + 4\delta cost(S, Z_0)
    \end{split}
\end{align}

\end{proof}

\begin{theorem}
 Algorithm  \ref{alg:ls} is an $ O(1)$-approximation algorithm for any fixed $0 < \epsilon \leq \frac{1}{4}$
\end{theorem}

 \begin{proof}
 \begin{align*}
 \begin{split}
         \sum_{(a,b) \in P} (cost(a, b) - cost(S, Z_0)) \\ \geq - \epsilon cost(S, Z_0).
         \biggl(8 + \frac{16}{\delta} \biggr) cost(S^*, Z^*) \\ +4 \delta cost(S, Z_0) + cost(S^*, Z^*) - \biggl(1 - \frac{\epsilon}{k} \biggr) cost(S, Z_0) \\ \geq -\epsilon cost(S, Z_0)
         \end{split}
 \end{align*}
Combining these inequalities gives
$\frac{9 + \frac{16}{\delta}}{ 1 - 4 \delta - \epsilon - \epsilon/k} cost(S^*, Z^*) \geq cost(S, Z_0).$
 \end{proof}

When $\delta = \frac{(2 \sqrt{79} -16)}{9}$ is small, the algorithm is shown to yield a $274$-approximation for the k-means objective.

\end{document}